\documentclass[final, nomarks]{dmtcs-episciences}

\usepackage[utf8]{inputenc}
\usepackage{subfigure}

\usepackage[round]{natbib}

\author{Michael Lampis\thanks{The author was supported by ANR JCJC projects ``S-EX-AP-PE-AL'' (ANR-21-CE48-0022) and ``ASSK'' (ANR-18-CE40-0025-01).}}
\title{Minimum Stable Cut and Treewidth}
\affiliation{ Université Paris-Dauphine, PSL University, CNRS, LAMSADE, Paris, France }

\publicationdata{vol. 28:2}{2026}{18}{10.46298/dmtcs.10900}{2023-02-02; 2023-02-02; 2025-08-21}{2026-03-09}

\usepackage{tikz}

\usepackage{wrapfig}

\newcommand{\msc}{\textsc{Min Stable Cut}}
\newcommand{\umsc}{\textsc{Unweighted Min Stable Cut}}
\newcommand{\tw}{\mathrm{tw}}
\newcommand{\pw}{\mathrm{pw}}

\newtheorem{theorem}{Theorem}
\newtheorem{remark}[theorem]{Remark}
\newtheorem{claim}[theorem]{Claim}
\newtheorem{lemma}[theorem]{Lemma}
\newtheorem{corollary}[theorem]{Corollary}

\begin{document}
\maketitle
\begin{abstract}

A stable or locally-optimal cut of a graph is a cut whose weight cannot be
increased by changing the side of a single vertex.  Equivalently, a cut is
stable if all vertices have the (weighted) majority of their neighbors on the
other side.  Finding a stable cut is a prototypical PLS-complete problem that
has been studied in the context of local search and of algorithmic game theory.

In this paper we study \msc, the problem of finding a stable cut of minimum
weight, which is closely related to the Price of Anarchy of the \textsc{Max
Cut} game.  Since this problem is NP-hard, we study its complexity on graphs of
low treewidth, low degree, or both. We begin by showing that the problem
remains weakly NP-hard on severely restricted trees, so bounding treewidth
alone cannot make it tractable. We match this hardness with a pseudo-polynomial
DP algorithm solving the problem in time $(\Delta\cdot W)^{O(\tw)}n^{O(1)}$,
where $\tw$ is the treewidth, $\Delta$ the maximum degree, and $W$ the maximum
weight.  On the other hand, bounding $\Delta$ is also not enough, as the
problem is NP-hard for unweighted graphs of bounded degree.  We therefore
parameterize \msc\ by both $\tw$ and $\Delta$ and obtain an FPT algorithm
running in time $2^{O(\Delta\tw)}(n+\log W)^{O(1)}$.  Our main result for the
weighted problem is to provide a reduction showing that both aforementioned
algorithms are essentially optimal, even if we replace treewidth by pathwidth:
if there exists an algorithm running in $(nW)^{o(\pw)}$ or
$2^{o(\Delta\pw)}(n+\log W)^{O(1)}$, then the ETH is false. Complementing this,
we show that we can, however, obtain an FPT \emph{approximation scheme}
parameterized by treewidth, if we consider almost-stable solutions, that is,
solutions where no single vertex can unilaterally increase the weight of its
incident cut edges by more than a factor of $(1+\varepsilon)$.

Motivated by these mostly negative results, we consider \umsc. Here our results
already imply a much faster exact algorithm running in time
$\Delta^{O(\tw)}n^{O(1)}$.  We show that this is also probably essentially
optimal: an algorithm running in $n^{o(\pw)}$ would contradict the ETH. 
\end{abstract}

\section{Introduction}

In this paper we study problems related to \emph{stable cuts} in graphs. A
\emph{stable} cut of an edge-weighted graph $G=(V,E)$ is a partition of $V$
into two sets $V_0, V_1$ that satisfies the following property: for each
$i\in\{0,1\}$ and $v\in V_i$, the total weight of edges incident on $v$ whose
other endpoint is in $V_{1-i}$ is at least half the total weight of all edges
incident on $v$. In other words, a cut is stable if all vertices have the
(weighted) majority of their incident edges cut.

The notion of stable cuts has been very widely studied from two different
points of view.  First, in the context of local search, a stable cut is a
locally optimal cut: switching the side of any single vertex cannot increase
the total weight of the cut. Hence, stable cuts have been studied with the aim
to further our understanding of the basic local search heuristic for
\textsc{Max Cut}.  Second, in the context of algorithmic game theory a
\textsc{Max Cut} game has often been considered, where each vertex is an agent
whose utility is the total weight of edges connecting it to the other side. In
this game, a stable cut corresponds exactly to the notion of a Nash
equilibrium, that is, a state where no agent has an incentive to change her
choice. The complexity of producing a Nash stable or locally optimal cut of a
given edge-weighted graph has been heavily studied under the name \textsc{Local
Max Cut}. The problem is known to be PLS-complete, under various restrictions
(we give detailed references below). 

In this paper we focus on a different but closely related optimization problem:
given an edge-weighted graph we would like to produce a stable cut \emph{of
minimum total weight}. We call this problem \msc. In addition to being a fairly
natural problem on its own, we believe that \msc\ is interesting from the
perspective of both local search and algorithmic game theory.  In the context
of local search, \msc\ is the problem of bounding the
performance of the local search heuristic on a particular instance. It is
folklore (and easy to see) that in general there exist graphs where the
smallest stable cut has size half the maximum cut (e.g.  consider a $C_4$)
and this is tight since any stable cut must cut at least half the total edge
weight.  However, for most graphs this bound is far from tight. \msc\ therefore
essentially asks to estimate the ratio between the largest and
smallest stable cut for a given specific instance.  Similarly, in the context
of algorithmic game theory, solving \msc\ is essentially equivalent to
calculating the Price of Anarchy of the \textsc{Max Cut} game on the given
instance, that is, the ratio between the smallest stable cut and the maximum
cut. Since we will mostly focus on cases where \textsc{Max Cut}
is tractable, \msc\ can, therefore, be seen as the problem of computing either
the approximation ratio of local search or the Price of Anarchy of the
\textsc{Max Cut} game on a given graph.

\subparagraph*{Our results} It appears that little is currently known about the
complexity of \msc. However, since finding a (not necessarily minimum) stable
cut is PLS-complete, finding the minimum such cut would be expected to be hard.
Our focus is therefore to study the parameterized complexity of \msc\ using
structural parameters such as treewidth and the maximum degree of the input
graph\footnote{We assume familiarity with the basics of parameterized
complexity as given in standard textbooks \cite{CyganFKLMPPS15}.}. Our results
are the following.

\begin{itemize}

\item First, we show that bounding only one of the two mentioned parameters is
not sufficient to render the problem tractable. This is not surprising for the
maximum degree $\Delta$, where a reduction from \textsc{Max Cut} allows us to
show the problem is NP-hard for $\Delta\le 6$ even in the unweighted case
(Theorem \ref{thm:npharddegree}). It is, however, somewhat more disappointing
that bounded treewidth also does not help, as the problem remains weakly
NP-hard on trees of diameter $4$ (Theorem \ref{thm:nphtrees}) and bipartite
graphs of vertex cover $2$ (Theorem \ref{thm:nphardvc}).

\item These hardness results point to two directions for obtaining algorithms
for \msc: first, since the problem is ``only'' weakly NP-hard for bounded
treewidth one could hope to obtain a pseudo-polynomial time algorithm in this
case.  We show that this is indeed possible and the problem is solvable in time
$(\Delta\cdot W)^{O(\tw)}n^{O(1)}$, where $W$ is the maximum edge weight
(Theorem \ref{thm:algpseudo}). Second, one may hope to obtain an FPT algorithm
when both $\tw$ and $\Delta$ are parameters. We show that this is also possible
and obtain an algorithm with complexity $2^{O(\Delta\tw)}(n+\log W)^{O(1)}$
(Theorem \ref{thm:algdelta}).

\item These two algorithms lead to two further questions. First, can the
$(\Delta\cdot W)^{O(\tw)}n^{O(1)}$ algorithm be improved to an FPT dependence
on $\tw$, that is, to running time $f(\tw)(nW)^{O(1)}$? And second, can the
$2^{\Delta\tw}$ parameter dependence of the FPT algorithm be improved, for
example to $2^{O(\Delta+\tw)}$ or even $\Delta^{O(\tw)}$? We show that the
answer to both questions is negative, even if we replace treewidth with
pathwidth: under the ETH there is no algorithm running in $(nW)^{o(\pw)}$ or
$2^{o(\Delta\tw)}(n+\log W)^{O(1)}$ (Theorem \ref{thm:eth1}).

\item Complementing the above, we show that the problem does become FPT by
treewidth alone if we allow the notion of approximation to be used in the
concept of stability: there exists an algorithm which, for any $\varepsilon>0$,
runs in time $(\tw/\varepsilon)^{O(\tw)}(n+\log W)^{O(1)}$ and produces a cut
with the following properties: all vertices are $(1+\varepsilon)$-stable, that is,
no vertex can unilaterally increase its incident cut weight by more than a
factor of $(1+\varepsilon)$; the cut has weight at most equal to that of the
minimum stable cut.

\item Finally, motivated by the above mostly negative results, we also consider
\umsc, the restriction of the problem where all edge weights are uniform.  Our
previous results give a much faster algorithm with parameter dependence
$\Delta^{O(\tw)}$, rather than $2^{\Delta\tw}$ (Corollary \ref{cor:algdelta2}).
However, this poses the natural question if in this case the problem finally
becomes FPT by treewidth alone.  Our main result in this part is to answer this
question in the negative and show that, under the ETH, \umsc\ cannot be solved
in time $n^{o(\pw)}$ (Theorem \ref{thm:hard2}).

\end{itemize}

Taken together, our results paint a detailed picture of the complexity of \msc\
parameterized by $\tw$ and $\Delta$. All our exact algorithms (Theorems
\ref{thm:algpseudo}, \ref{thm:algdelta}) are obtained using standard dynamic
programming on tree decompositions, the only minor complication being that for
Theorem \ref{thm:algdelta} we edit the decomposition to make sure that for each
vertex some bag contains all of its neighborhood (this helps us verify that a
cut is stable). The main technical challenge is in proving our complexity lower
bounds. It is therefore perhaps somewhat surprising that the lower bounds turn
out to be essentially tight, as this indicates that for \msc\ and \umsc, the
straightforward DP algorithms are essentially optimal, if one wants to solve
the problem exactly. 

For the approximation algorithm, we rely on two rounding techniques: one is a
rounding step similar to the one that gives an FPTAS for \textsc{Knapsack} by
truncating weights so that the maximum weight is polynomially bounded.
However, \msc\ is more complicated than \textsc{Knapsack}, as an edge which is
light for one of its endpoints may be heavy for the other. We therefore define
a more general version of the problem, allowing us to decouple the contribution
each edge makes to the stability of each endpoint.  This helps us bound the
largest stability-weight by a polynomial, but is still not sufficient to obtain
an FPT algorithm, as the lower bound of Theorem \ref{thm:eth1} applies to
polynomially bounded weights.  We then go on to apply a technique introduced in
\cite{Lampis14} (see also
\cite{AngelBEL18,BelmonteLM20,KatsikarelisLP19,KatsikarelisLP22,KulikS26,LampisMNOV025}) which allows
us to obtain FPT approximation algorithms for problems which are W-hard by
treewidth by applying a different notion of rounding to the dynamic program.
This allows us to produce a solution that is simultaneously of optimal weight
(compared to the best stable solution) and almost-stable, using essentially the
same algorithm as in Theorem \ref{thm:algpseudo}. However, it is worth noting
that in general there is no obvious way to transform almost-stable solutions to
stable solutions \cite{BhalgatCK10,CaragiannisFGS15}, so our algorithm is not
immediately sufficient to obtain an FPT approximation for \msc\ if we insist on
obtaining a cut which is exactly stable.

\subparagraph*{Related work} From the point of view of local search algorithms,
there is an extensive literature on the \textsc{Local Max Cut} problem, which
asks us to find a stable cut (of any size). The problem has long been known to
be PLS-complete \cite{JohnsonPY88, SchafferY91}. It remains PLS-complete for
graphs of maximum degree $5$ \cite{ElsasserT11}, but becomes polynomial-time
solvable for graphs of maximum degree $3$ \cite{Loebl91,Poljak95}. The problem
remains PLS-complete if weights are assigned to vertices, instead of edges, and
the weight of an edge is defined simply as the product of the weights of its
endpoints \cite{FotakisKLMPS20}. Even though the problem is PLS-complete, it
has long been observed that local search quickly finds a stable solution in
most practical instances. One theoretical explanation for this phenomenon was
given in a recent line of work which showed that \textsc{Local Max Cut} has
quasi-polynomial time smoothed complexity
\cite{AngelBPW17,BibakCC21,ChenGVYZ20,EtscheidR17}. \textsc{Local Max Cut} is
of course polynomial time solvable if all weights are polynomially bounded in
$n$, as local improvements always increase the size of the cut.

In algorithmic game theory much work has been done on the complexity of
computing Nash equilibria for the cut game and the closely related \emph{party
affiliation game}, in which players, represented by vertices, have to pick one
of two parties and edge weights indicate how much two players gain if they are
in the same party
\cite{AwerbuchAEMS08,BalcanBM09,ChristodoulouMS12,FabrikantPT04,GourvesM09}.
Note that for general graphical games finding an equilibrium is PPAD-hard
on trees of constant pathwidth \cite{ElkindGG06}.  Because computing a
stable solution is generally intractable, approximate equilibria
have also been considered \cite{BhalgatCK10,CaragiannisFGS15}. Note that the
notion of approximate equilibrium corresponds exactly to the 
approximation guarantee given by Theorem \ref{thm:algapprox}, but unlike the
cited works, Theorem \ref{thm:algapprox} produces a solution that is both
approximately stable and as good as the optimal. Another problem motivated by
algorithmic game theory that is close to the problem of this paper is the study
of \emph{hedonic games}, where vertices are agents who form coalitions, and
edge weights express how much a player gains by being in the same coalition as
one of her neighbors (note that in such games, the number of coalitions is not
a priori restricted to two).  Finding a Nash stable configuration in such a
game is hard, and the complexity of this problem parameterized by treewidth and
$\Delta$ has recently been settled \cite{Peters16a,HanakaKL26,HanakaL22}

The problem we consider in this paper is more closely related to the problem of
computing the \emph{worst} (or best) Nash equilibrium, which in turn is closely
linked to the notion of Price of Anarchy. For most problems in
algorithmic game theory this type of question is usually NP-hard
\cite{BiloM21,ConitzerS08,ElkindGG07,FotakisKKMS09,gilboa1989nash,GrecoS09,SchoenebeckV12}
and hard to approximate
\cite{AustrinBC13,BravermanKW15,DeligkasFS18,HazanK11,MinderV09}. Even though
these results show that finding a Nash equilibrium that maximizes an objective
function is NP-hard under various restrictions (e.g. graphical games of bounded
degree), to the best of our knowledge the complexity of finding the worst
equilibrium of the \textsc{Max Cut} game (which corresponds to the \msc\
problem of this paper) has not been considered.

Finally, another topic that has recently attracted attention in the literature
is that of MinMax and MaxMin versions of standard optimization problems, where
we search the worst solution which cannot be improved using a simple local
search heuristic. The motivation behind this line of research is to provide
bounds and a refined analysis of such basic heuristics. Problems that have been
considered under this lens are \textsc{Max Min Dominating Set}
\cite{Bazgan2018,DubloisLP22}, \textsc{Max Min Vertex Cover}
\cite{Bonnet2018,Zehavi2017}, \textsc{Max Min Separator} \cite{Hanaka2019},
\textsc{Max Min Cut} \cite{EtoHKK2019}, \textsc{Min Max Knapsack}
\cite{ArkinBMS03,Furini2017,GourvesMP13}, \textsc{Max Min Edge Cover}
\cite{KhoshkhahGMS20}, \textsc{Max Min Feedback Vertex Set} \cite{DubloisHGLM22}. Some
problems in this area also arise naturally in other forms and have been
extensively studied, such as \textsc{Min Max Matching} (also known as
\textsc{Edge Dominating Set} \cite{IwaideN16}) and \textsc{Grundy Coloring},
which can be seen as a \textsc{Max Min} version of \textsc{Coloring}
\cite{AboulkerBKS23,BelmonteKLMO22}.

\section{Definitions -- Preliminaries}

We generally use standard graph-theoretic notation and consider edge-weighted
graphs, that is, graphs $G=(V,E)$ supplied with a weight function $w:E\to
\mathbb{N}$.  For a vertex $v\in V$, The weighted degree of a vertex $v\in V$
is $d_w(v)=\sum_{uv\in E} w(uv)$. A cut of a graph is a partition of $V$ into
$V_0, V_1$. A cut is \emph{stable} for vertex $v\in V_i$ if $\sum_{vu\in E\land
u\in V_{1-i}} w(vu) \ge \frac{d_w(v)}{2}$, that is, if the total weight of
edges incident on $v$ crossing the cut is at least half the weighted degree of
$v$. In the \msc\ problem we are given an edge-weighted graph and are looking
for a cut that is stable for all vertices that minimizes the sum of weights of
cut edges (that is, edges with endpoints on both sides of the cut). In \umsc\
we restrict the problem so that the $w$ function returns $1$ for all edges.
When describing stable cuts we will sometimes say that we ``assign'' value $0$
(or $1$) to a vertex; by this we mean that we place this vertex in $V_0$ (or
$V_1$ respectively).

For the definitions of treewidth, pathwidth, and the related (nice)
decompositions we refer to \cite{CyganFKLMPPS15}. We will use as a complexity
assumption the Exponential Time Hypothesis (ETH) \cite{ImpagliazzoPZ01} which
states that there exists a constant $c>1$ such that \textsc{3-SAT} with $n$
variables and $m$ clauses cannot be solved in time $c^{n+m}$. In fact, we will
use the slightly weaker and simpler form of the ETH which states that
\textsc{3-SAT} cannot be solved in time $2^{o(n+m)}$.

\section{Weighted Min Stable Cut}

In this section we present our results on exact algorithms for (weighted) \msc.
We begin with some basic NP-hardness reductions in Section \ref{sec:nphard},
which establish that the problem remains (weakly) NP-hard when either the
treewidth or the maximum degree  are bounded. These set the stage for two
algorithms, given in Section \ref{sec:algs}, solving the problem in
pseudo-polynomial time for constant treewidth; and in FPT time parameterized by
$\tw+\Delta$. In Section \ref{sec:hardeth} we present a more fine-grained
hardness argument, based on the ETH, which shows that the dependence on $\tw$
and $\Delta$ of our two algorithms is essentially optimal.

\subsection{Basic Hardness Proofs}\label{sec:nphard}

\begin{theorem}\label{thm:nphtrees} \msc\ is weakly NP-hard on trees of
diameter $4$.  \end{theorem}

\begin{proof} We describe a reduction from \textsc{Partition}. Recall that in
this problem we are given $n$ positive integers $x_1,\ldots, x_n$ such that
$\sum_{i=1}^n x_i = 2B$ and are asked if there exists $S\subseteq [n]$ such
that $\sum_{i\in S} x_i = B$.  We construct a star with $n$ leaves and
subdivide every edge once. For each $i\in [n]$ we select a distinct leaf of the
tree and set the weight of both edges in the path from the center to this leaf
to $x_i$. We claim that the graph has a stable cut of weight $3B$ if and only
if there is a partition of $x_1,\ldots, x_n$ into two sets with the same sum.

For the first direction, suppose $S\subseteq [n]$ is such that $\sum_{i\in S}
x_i = B$. For each $i\in S$ we select a degree two vertex of the tree whose
incident edges have weight $x_i$ and assign it value $1$. We assign all other
degree two vertices value $0$ and assign to all leaves the opposite of the
value of their neighbor. We give the center value $0$. This partition is stable
as the center has edge weight exactly $B$ towards each side, and all degree two
vertices have a leaf attached that is placed on the other side and contributes
half their total incident weight. The total weight cut is $2B$ from edges
incident on leaves, plus $B$ from half the weight incident on the center.

For the converse direction, observe that in any stable solution all edges
incident on leaves are cut, contributing a weight of $2B$. As a result, in a
stable cut of size $3B$, the weight of cut edges incident on the center is at
most $B$. However, this weight is also at least $B$, since the edge weight
incident on the center is $2B$. We conclude that the neighborhood of the center
must be perfectly balanced. From this we can infer a solution to the
\textsc{Partition} instance.  \end{proof}

\begin{remark}\label{rem:rem} Theorem \ref{thm:nphtrees} is tight, because \msc\ is trivial on
trees of diameter at most $3$. \end{remark}

\begin{proof}
A tree of diameter at most $3$ must be either a star, in which case there is
only one feasible solution (up to symmetry); or a double-star, that is a graph
produced by taking two stars and connecting their centers. In the latter case,
the optimal solution is always to place the two centers on the same side if
this is feasible (as otherwise all edges are cut).  \end{proof}

\begin{theorem}\label{thm:nphardvc} \msc\ is weakly NP-hard on bipartite graphs
with vertex cover $2$.  \end{theorem}

\begin{proof}
We present a reduction from \textsc{Partition} similar to that of
Theorem~\ref{thm:nphtrees}. Given an instance with values $x_1,\ldots, x_n$ we
construct a bipartite graph $K_{2,n}$. To ease presentation, we will call the
part of $K_{2,n}$ that contains two vertices the ``left'' part, and the part
that contains the remaining $n$ vertices the ``right'' part. For each $i\in
[n]$ we select a vertex of the right part and set the weight of both its
incident edges to $x_i$. We claim that this graph has a stable cut of weight
$2B$ if and only if the original instance is a Yes instance.

If there is a partition $S\subseteq [n]$ such that $\sum_{i\in S} x_i = B$, we
select the corresponding vertices of the right part and assign to them $0$; we
assign $1$ to the other vertices of the right part; we assign $0$ to one vertex
of the left part and $1$ to the other. This partition is stable, as all
vertices have completely balanced neighborhoods. Furthermore, the weight of the
cut is $2B$.

For the other direction, observe that if both vertices of the left part of
$K_{2,n}$ are on the same side of the partition, then all edges will be cut,
giving weight $4B$. So a stable partition of weight $2B$ must place these two
vertices on different sides. However, these vertices have the same neighbors
(with the same edge weights), so if both are stable, their neighborhood must be
properly balanced. From this we can infer a solution to the \textsc{Partition}
instance.  \end{proof}

\begin{theorem}\label{thm:npharddegree} \umsc\ is strongly NP-hard and APX-hard
on bipartite graphs of maximum degree $6$. \end{theorem}

\begin{proof} 
We give a reduction from \textsc{Max Cut} on graphs of maximum degree $3$,
which is known to be APX-hard \cite{BermanK99}. Given an instance $G=(V,E)$ of
\textsc{Max Cut} we sub-divide each edge of $E$ once, and we attach three
leaves to each vertex of $V$. The new graph therefore has $3|V|+2|E|$ edges. We
claim that if the original instance has a cut of size at least $k$ then the new
instance has a stable cut of size at most $3|V|+2|E|-k$.

For one direction, suppose we have a cut of $G$ of size $k$ which partitions
$V$ into $V_0, V_1$. We use the same partition of $V$ for the new instance. For
each leaf, we assign it a value opposite of that of its neighbor. For each
degree two vertex which was produced when sub-dividing an edge of $E$ we give
it a value that is opposite to that of at least one of its neighbors. Observe
that this cut is stable: all leaves are stable; all vertices produced in
sub-divisions have degree two and at least one neighbor on the other side; and
all vertices of $V$ are adjacent to three leaves on the other side and at most
three other vertices (since $G$ is subcubic). The edges cut are: $3|V|$ edges
incident on leaves; $2$ edges for each edge of $E$ whose endpoints are on the
same side; $1$ edge for each cut edge of $E$. This gives $3|V|+2|E|-k$ edges
cut overall.

Conversely, suppose we have a stable cut of the new graph. We use the same cut
in $G$ and claim that it must cut at least $k$ edges. Indeed, in the new graph
any stable cut must cut all $3|V|$ edges incident on leaves, and at least one
of the two edges incident on each degree two vertex. Furthermore, if
$e=(u,v)\in E$ and $u,v$ are on the same side of the cut, then both edges in
the sub-divided edge $e$ must be cut. We conclude that there must be at least
$k$ edges of $G$ with endpoints on different sides of the cut.  \end{proof}

\subsection{Algorithms}\label{sec:algs}

\begin{theorem}\label{thm:algpseudo} There is an algorithm which, given an
instance of \msc\ with $n$ vertices, maximum weight $W$, maximum degree $\Delta$, and a tree
decomposition of width $\tw$, finds an optimal solution in time $(\Delta\cdot
W)^{O(\tw)}n^{O(1)}$. \end{theorem}

\begin{proof}
Our algorithm follows the standard dynamic programming method for treewidth. We
assume that we are given a nice tree decomposition of width $\tw$ for the input
graph $G=(V,E)$; this can be done without loss of generality, because it is
known than an arbitrary tree decomposition can be transformed into a nice tree
decomposition in polynomial time, without increasing its width
\cite{CyganFKLMPPS15}.  For each node $t$ of the decomposition, let
$B_t\subseteq V$ be the bag associated with $t$ and $B_t^{\downarrow}\subseteq
V$ the set of all vertices of $G$ which appear in bags in the sub-tree rooted
in $t$ (that is, the vertices which appear below $t$ in the decomposition). 

A partial solution for node $t$ is a partition of $B_t^\downarrow$ into
$V_0,V_1$. We define the signature of a partial solution in node $t$ as a tuple
of the following information: (i) a partition of $B_t$ into two sets, which
encodes the intersections of $B_t$ with $V_0, V_1$ (ii) for each $v\in B_t$ an
integer value in $\{0,\ldots, d_w(v) \}$, which encodes the total weight of its
incident edges whose other endpoint is in $B_t^{\downarrow}\setminus B_t$ and
on the same side of the cut as $v$.  In other words, a signature $\sigma$ of a
partial solution (that is, a partition of $B_t^\downarrow$ into $V_0, V_1$) in
node $t$ is a pair $\sigma=(S,f)$, where $S\subseteq B_t$ and $f:B_t\to
\mathbb{N}$. The intended meaning is that the partition $V_0,V_1$ has signature
$\sigma=(S,f)$ if and only if (i) $S=B_t\cap V_0$ and (ii) for each $v\in B_t$,
if $v\in V_i$, for $i\in\{0,1\}$, then the total weight of edges incident on
$v$ and on vertices of $V_i\cap (B_t^{\downarrow}\setminus B_t)$ is $f(v)$.

Our dynamic programming (DP) algorithm constructs for each node $B_t$, a table
$T_t$ indexed by the possible signatures $\sigma$ for node $t$. The table
associates with each signature $\sigma$  a value $T_t(\sigma)$, which is the
weight of the best (minimum) partial solution of $B_t^{\downarrow}$ that is (i)
consistent with the signature and (ii) stable for all vertices of
$B_t^{\downarrow}\setminus B_t$. If no such solution exists, the table sets
$T_t(\sigma)=\infty$. More formally, our algorithm will satisfy the following
two invariants for all nodes $t$ of the tree decomposition:

\begin{enumerate}

\item For all partitions $V_0,V_1$ of $B_t^\downarrow$ such that all vertices
$v\in B_t^\downarrow\setminus B_t$ are stable, we have the following: if
$\sigma$ is the signature corresponding to the partition $V_0,V_1$ and $c$ is
the total weight of edges cut in $V_0,V_1$ with at least one endpoint outside
of $B_t$, then $T_t(\sigma)\le c$.

\item For all signatures $\sigma$ such that $T_t(\sigma)\neq\infty$, there
exists a partition $V_0,V_1$ of $B_t^\downarrow$ satisfying the following: (i)
all vertices $v\in B_t^\downarrow\setminus B_t$ are stable (ii) the total
weight of edges cut in $V_0,V_1$ with at least one endpoint outside of $B_t$ is
$T_t(\sigma)$ (iii) the partition $V_0,V_1$ has signature $\sigma$.

\end{enumerate}

Informally, the two invariants above state that (i) for all partial solutions,
our DP tables will contain an entry that is at least as good as the partial
solution (ii) every entry of the DP tables will in fact correspond to a valid
solution (where validity is checked below the bag). We initialize all tables to
contain the value $\infty$ for all signatures and then process the
decomposition bottom-up starting from the leaves. When it is time to process a
node $t$ we will assume that the invariants hold for its children. Our
algorithm proceeds according to the following cases:

\begin{itemize}

\item $t$ is a Leaf node. Here there is only one possible signature
($\sigma=(\emptyset,\emptyset)$), so we set $T_t(\sigma)=0$. 

\item $t$ is an Introduce node, with child $t'$ and $B_t=B_{t'}\cup\{v\}$. For
every signature $\sigma=(S,f)$ such that $T_{t'}(\sigma)\neq\infty$, we
construct two signatures $\sigma_1=(S,f')$ and $\sigma_2=(S\cup\{v\},f')$,
where $f'$ is the function that agrees with $f$ on $B_{t'}$ and has $f'(v)=0$.
We set $T_t(\sigma_1)=T_t(\sigma_2)=T_{t'}(\sigma)$.

\item $t$ is a Forget node, with child $t'$ and $B_t=B_{t'}\setminus\{v\}$.  We
consider each signature $\sigma=(S,f)$ such that $T_{t'}(\sigma)\neq\infty$.
For each such signature, we let $x$ be the total weight of edges $uv$ with
$u\in B_t$ and $|\{u,v\}\cap S|\in\{0,2\}$. If $x+f(v)>\frac{d_w(v)}{2}$ we
skip this signature.  Otherwise, let $\sigma'=(S\setminus\{v\},f')$, where $f'$
is defined as follows: if $|\{u,v\}\cap S|=1$ then $f'(u)=f(u)$; otherwise
$f'(u)=f(u)+w(uv)$.  We set $T_t(\sigma'):=\min\{T_t(\sigma'),
T_{t'}(\sigma)+y\}$, where $y$ is the total weight of edges $uv$ with $u\in
B_t$ and $|\{u,v\}\cap S|=1$.

\item $t$ is a Join node, with children $t_1,t_2$ and $B_t=B_{t_1}=B_{t_2}$.
For any two signatures $\sigma_1=(S,f_1)$ and $\sigma_2=(S,f_2)$ such that
$T_{t_1}(\sigma_1)\neq\infty$ and $T_{t_2}(\sigma_2)\neq\infty$ we construct a
signature $\sigma=(S,f_1+f_2)$ and set $T_t(\sigma):=\min\{T_t(\sigma),
T_{t_1}(\sigma_1)+T_{t_2}(\sigma_2)\}$.

\end{itemize}

\begin{claim} After executing the algorithm above, the invariants are satisfied
for all nodes $t$ of the given nice tree decomposition. \end{claim}

\begin{proof} We prove the claim by induction. For the base case we
consider a Leaf node and it is clear that the invariants hold. 

For an Introduce node $t$, any partition $V_0,V_1$ of $B_t^\downarrow$ gives a
partition of $B_{t'}^\downarrow$, which corresponds to a signature $\sigma'$,
such that $T_{t'}(\sigma')\neq\infty$ (by inductive hypothesis and the first
invariant).  Our algorithm will therefore produce two signatures
$\sigma_1,\sigma_2$ represented in $T_t$, one of which agrees with the
partition, because $v$ has no neighbors in $B_t^\downarrow\setminus B_t$.
Similarly, each such produced signature $\sigma_1,\sigma_2$ is constructed
based on a signature from $T_{t'}$, so the second invariant is also satisfied.

For a Forget node $t$, consider a partition $V_0,V_1$ of $B_t^\downarrow$. If
the partition $V_0,V_1$ has the property that all $v\in B_t^\downarrow\setminus
B_t$ are stable, then the signature $\sigma=(S,f)$ of $V_0,V_1$ in $t'$ was not
skipped by our algorithm. Our algorithm then constructed a signature $\sigma'$
by updating for each $u\in B_t$ its total weight to its side of the partition
towards vertices of $B_t^\downarrow\setminus B_t$. For the second invariant, it
is not hard to see that if our algorithm gives non-infinite value to a
signature $\sigma$, this is because a corresponding signature (and by induction
a partial solution) exists for $B_{t'}$.

For a Join node $t$, any partition $V_0,V_1$ of $B_t^\downarrow$ with signature
$\sigma=(S,f)$ corresponds to signatures $\sigma_1=(S,f_1), \sigma_2=(S,f_2)$
in $t_1,t_2$ respectively and by definition it must be the case that
$f=f_1+f_2$.  \end{proof}

What remains is to bound the running time of our algorithm, which is polynomial
in the size of the decomposition and the sizes of the DP tables.  The total
number of possible signatures is at most $2^{|B_t|} (\max d_w(v))^{|B_t|} \le
O(2^{\tw} (\Delta\cdot W)^{\tw+1})$, because $d_w(v)$ is always upper-bounded
by $\Delta\cdot W$.  \end{proof}

\begin{theorem}\label{thm:algdelta} There is an algorithm which, given an
instance of \msc\ with $n$ vertices, maximum weight $W$, maximum degree
$\Delta$  and a tree decomposition of width $\tw$, finds an optimal solution in
time $2^{O(\Delta\tw)}(n+\log W)^{O(1)}$. \end{theorem}

\begin{proof}
We describe an algorithm which works in a way similar to the standard
algorithm for \textsc{Max Cut} parameterized by
treewidth \cite{BodlaenderJ00}, except that we work in a tree decomposition
that is essentially a decomposition of the square of $G$\footnote{The square of
$G$ is the graph with the same vertices as $G$ where $u,v$ are adjacent if they
are at distance at most $2$ in $G$.}.  More precisely, before we begin, we do
the following: for each $v\in V$ we add to every bag of the decomposition that
contains $v$ all the vertices of $N(v)$.  It is not hard to see that we now
have a decomposition of width at most $(\Delta+1)(\tw+1)$ and also that the new
decomposition is still a valid tree decomposition.  Crucially, we now also have
the following property: for each $v\in V$ there exists at least one bag of the
decomposition that contains all of $N[v]$.

The algorithm now performs dynamic programming by storing for each bag the
value of the best solution for each partition of $B_t$. As a result, the size
of the DP table is $2^{O(\Delta\tw)}$. The only difference with the standard
\textsc{Max Cut} algorithm (beyond the fact that we are looking for a cut of
minimum weight) is that when we consider a bag that contains all of $N[v]$, for
some $v\in V$, we discard all partitions which are unstable for $v$. Since the
bag contains all of $N[v]$, this can be checked in time polynomial in $n$ and
$\log W$ (assuming weights are given in binary).  \end{proof}

\subsection{Tight ETH-based Hardness}\label{sec:hardeth}

Our goal in this section is to show that the parameter dependence
$2^{O(\Delta\tw)}$ given by the algorithm of Theorem~\ref{thm:algdelta} is essentially
optimal. Before we proceed to the reduction that will establish this (under the
ETH), it is worth thinking a bit more clearly about the challenge of proving
such a lower bound involving two distinct parameters. Recall that, by now, the
techniques for proving that a problem does not admit, say, an FPT algorithm
with parameter dependence $\tw^{o(\tw)}$, or $2^{o(\tw^2)}$ are
well-understood: we need to start a reduction from a \textsc{3-SAT} instance
with $n$ variables and produce an equivalent instance of our problem where the
treewidth is bounded by $O(n/\log n)$ or $O(\sqrt{n})$ respectively. Then, an
algorithm with such a ``too fast'' running time would violate the ETH.

Given the above, one may be tempted to prove the optimality of the algorithm of
Theorem~\ref{thm:algdelta} by showing that no algorithm can solve the problem in time
$2^{o(\Delta\tw)}n^{O(1)}$. For example, we could come up with a reduction from
$n$-variable \textsc{3-SAT} producing instances with $\tw=O(\sqrt{n})$ and
$\Delta=O(\sqrt{n})$, in which case $2^{o(\Delta\tw)}=2^{o(n)}$, therefore such
an algorithm would violate the ETH. Unfortunately, though such a statement
would be technically correct, it would not satisfactorily establish that the
algorithm of Theorem~\ref{thm:algdelta} is ``best possible''. The reason is that,
unlike functions of a single variable, increasing functions of two variables
are not totally ordered. Concretely, consider the possibility of an algorithm
with running time $2^{\Delta^2+\tw}n^{O(1)}$. Such an algorithm would not be
ruled out by the reduction we hypothesized above (which sets
$\Delta=O(\sqrt{n})$), but it would be incomparable to the algorithm of
Theorem~\ref{thm:algdelta}, in the sense that one could be much faster than the other,
depending on whether $\Delta$ is larger than $\tw$. In other words, if our goal
is to convincingly argue that the ``correct'' complexity is exponential in
\emph{the product} $\Delta\cdot\tw$, we also need to be able to rule out
algorithms which sacrifice a bit the dependence on one parameter to obtain a
better dependence on the other.

In order to achieve this more comprehensive kind of hardness result, we will
therefore present a parametric kind of reduction, which starts from an instance
of $n$-variable \textsc{3-SAT} and, for \emph{any} desired (reasonable) function
$\Delta(n)$ produces an equivalent instance of maximum degree $\Delta(n)$ and
treewidth $O(n/\Delta(n))$. Because our reduction will work for any
$\Delta(n)$, it will allow us to establish that the correct exponent is at
least $\Delta\cdot \tw$, no matter the relation between $\Delta$ and $\tw$ in
the input instance.

We now give such a parametric reduction from \textsc{3-Set Splitting} to \msc\
whose main properties are laid out in Lemma \ref{lem:hard}. This reduction
gives the lower bound of Theorem \ref{thm:eth1}. Recall that the \textsc{3-Set
Splitting} problem is the following: we are given a hypergraph where hyperedges
have arity at most $3$ and the question is to decide whether there is a
partition of the vertex set such that all hyperedges intersect both sides of
the partition.

\begin{lemma}\label{lem:hard} There is a polynomial-time algorithm which, given
a \textsc{3-Set Splitting} instance $H=(V,E)$ with $n$ elements, and a positive
integer $\delta$, produces a \msc\ instance $G$ with the following properties:
(i) $G$ is a Yes instance if and only if $H$ is a Yes instance; (ii) if
$\Delta$ is the maximum degree of $G$ and $\pw$ its pathwidth, then $\Delta =
O(\delta)$ and $\pw=O(n/\delta)$; (iii) the maximum weight of $G$ is
$W=O(2^{\Delta})$.  \end{lemma}

\begin{figure}[h]
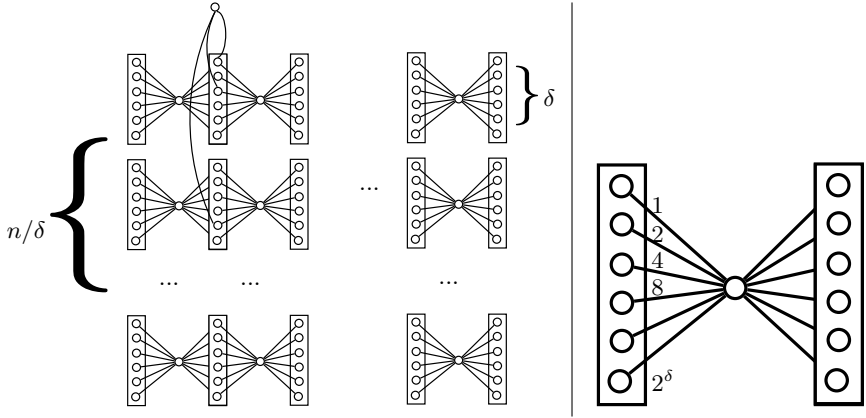


\begin{tabular}{l|r}
\input{red1}&

\input{red2}
\end{tabular}
\caption{Sketch of the construction of Lemma \ref{lem:hard}. On the left, the general architecture: $m$ columns, each with $n$ vertices, partitioned into groups of size $\log n$. On each column we add a checker vertex (on top). Between the same groups of consecutive columns we add propagator vertices. On the right, more details about the exponentially increasing weights of edges incident on propagators.}\label{fig:red1}
\end{figure}

\begin{proof}  
Let $H=(V,E)$ be the given \textsc{3-Set Splitting} instance, $V=\{v_0,\ldots,
v_{n-1}\}$ and suppose that $E$ contains $e_2$ sets of size $2$ and $e_3$ sets
of size $3$, where $|E|=e_2+e_3$ will be denoted by $m$.  Assume without loss
of generality that $n$ is a power of $2$ (otherwise add some dummy elements to
$V$). We construct a graph by first making $m$ copies of $V$, call them $V_j,
j\in [m]$ and label their vertices as $V_j = \{ v_{(i,j)}\ |\ i\in \{0,\ldots,
n-1\}\}$. Intuitively, the vertices $\{ v_{(i,j)}\ |\ j\in [m]\}$ are all meant
to represent the element $v_i$ of $H$.  We now add to the graph the following:

\begin{enumerate}

\item Checkers: Suppose that the $j$-th set of $E$ contains elements $v_{i_1},
v_{i_2}, v_{i_3}$. Then we construct a vertex $c_j$ and connect it to
$v_{(i_1,j)}, v_{(i_2,j)}, v_{(i_3,j)}$ with edges of weight $1$. If the $j$-th
set has size two, we do the same (ignoring $v_{i_3}$). 

\item Propagators: For each $j\in [m-1]$ we construct $\rho = \lceil n/\delta
\rceil$ vertices labeled $p_{(i',j)}, i'\in \{0,\ldots,\rho-1\}$. Each
$p_{(i',j)}$ is connected to (at most) $\delta$ vertices of $V_j$ and $\delta$
vertices of $V_{j+1}$ with edges of exponentially increasing weight.
Specifically, for $i'\in \{ 0,\ldots, \rho-1\}, \ell\in \{0,\ldots,\delta-1\}$,
we connect $p_{(i',j)}$ to $v_{(i'\delta+\ell,j)}$ and to $v_{(i'\delta +
\ell,j+1)}$ (if they exist) with an edge of weight $2^{\ell}$.

\item Stabilizers: For each $j\in [m], i\in \{0,\ldots,n-1\}$ we attach to
$v_{(i,j)}$ a leaf. The edge connecting this leaf to $v_{(i,j)}$ has weight
$3\cdot 2^{(i\bmod \delta)}$.

\end{enumerate}

This completes the construction of the graph. Let $L_w$ be the total weight of
edges incident on leaves and $P$ be the total weight of edges incident on
Propagator vertices $p_{(i,j)}$. We set $B=L_w+\frac{P}{2}+e_2+2e_3$ and claim
that the new instance has a stable cut of weight $B$ if and only if $H$ can be
split.

For the forward direction, suppose that $H$ can be split by the partition of
$V$ into $L, R=V\setminus L$. We assign the following values for our new
instance: for each $j\in [m]$ odd, we set $v_{(i,j)}$ to value $0$ if and only
if $v_i\in L$; for each $j\in [m]$ even, we set $v_{(i,j)}$ to value $0$ if and
only if $v_i\in R$. In other words, we use the same partition for all copies of
$V$, but flip the roles of $0,1$ between consecutive copies. We place leaves on
the opposite side from their neighbors and greedily assign values to all other
vertices of the graph to obtain a stable partition.  Observe that all vertices
$v_{(i,j)}$ are stable with the values we assigned, since the edge connecting
each such vertex to a leaf has weight at least half its total incident weight.
More specifically, a vertex $v_{(i,j)}$ has an edge of weight $3\cdot
2^{(i\bmod \delta)}$ connecting it to a leaf (and this edge is cut); at most
two edges of weight $2^{(i\bmod \delta)}$ each connecting it to Propagators;
and at most an edge of weight $1$ connecting it to a Checker.

In the partition we have, we observe that (i) all edges incident on leaves are
cut (total weight $L_w$) (ii) all Propagator vertices have balanced
neighborhoods, so exactly half of their incident weight is cut (total weight
$P/2$) (iii) since $L,R$ splits all sets of $E$, each checker vertex will have
exactly one neighbor on the same side (total weight $e_2+2e_3$). So, the total
weight of the cut is $B$.

For the converse direction, suppose we have a stable cut of size $B$ in the
constructed instance. Because of the stability condition, this solution must
cut all edges incident on leaves (total weight $L_w$); at least half of the
total weight of edges incident on Propagators (total weight $P/2$); and for
each checker vertex all its incident edges except at most one (total weight at
least $e_2+2e_3$).  We conclude that, in order to achieve weight $B$ all
previous bounds must be tight, that is, the cut must properly balance the
neighborhood of all Propagators and make sure that each Checker vertex has one
neighbor on its own side.

We now argue that because the neighborhood of each Propagator is balanced we
have for all $i\in\{0,\ldots, n-1\}, j\in [m-1]$ that $v_{(i,j)}, v_{(i,j+1)}$
are on different sides of the partition. To see this, suppose for contradiction
that for two such vertices this is not the case and to ease notation consider
the vertices $v_{(i\delta+\ell,j)}, v_{(i\delta+\ell,j+1)}$, where $0\le
\ell\le \delta-1$.  Among all such pairs select one that maximizes $\ell$.
Both vertices are connected to the Propagator $p_{(i , j)}$ with edges of
weight $2^{\ell}$. But now $p_{(i,j)}$ has strictly larger edge weight
connecting it to the side of the partition that contains $v_{(i\delta+\ell,j)}$
and $v_{(i\delta+\ell,j+1)}$ than to the other side because (i) for neighbors
of $p_{(i,j)}$ connected to it with edges of higher weight, the neighborhood of
$p_{(i,j)}$ is balanced by the maximality of $\ell$ (ii) the total weight of
all other edges is $2\cdot (2^{\ell-1}+2^{\ell-2}+\ldots+1) < 2\cdot 2^{\ell}$. 

We thus have that for all $i,j$, $v_{(i,j)}, v_{(i,j+1)}$ must be on different
sides, and therefore all $V_j$ are partitioned in the same way (except some
have the roles of $0$ and $1$ reversed). From this, we obtain a partition of
$V$.  To conclude this direction, we argue that this partition of $V$ must
split all sets. Indeed, if not, there will be a checker vertex such that all
its neighbors are on the same side, which, as we argued, means that the cut
must have weight strictly more than $B$.

Finally, let us show that the constructed instance has the claimed properties.
The maximum degree is $\Delta = 2\delta$ in the Propagators vertices (all other
vertices have degree at most $4$); the maximum weight is $O(2^{\delta}) =
O(2^{\Delta})$. Let us also consider the pathwidth of the constructed graph.
Let $G_j$ be the subgraph induced by $V_j$ and its attached leaves, the Checker
$c_j$, and all Propagators adjacent to $V_j$. We claim that we can build a path
decomposition of $G_j$ that contains all Propagators adjacent to $V_j$ in all
bags and has width $O(n/\delta)$. Indeed, if we place all the (at most $\lceil
2n/\delta\rceil$) Propagators and $c_j$ in all bags, we can delete them from
$G_j$, and all that is left is a union of isolated edges, which has pathwidth
$1$. Now, since the union of all $G_j$ covers all vertices and edges, we can
construct a path decomposition of the whole graph of width $O(n/\delta)$ by
gluing together the decompositions of each $G_j$, that is, by connecting the
last bag of the decomposition of $G_j$ to the first bag of the decomposition of
$G_{j+1}$. As claimed, the whole construction can be performed in polynomial
time, assuming that weights are written in binary.  \end{proof}

We can first obtain a ``standard'' lower bound, by invoking Lemma~\ref{lem:hard}
with $\delta=\log n$. This is sufficient to prove that the pseudopolynomial
algorithm of Theorem~\ref{thm:algpseudo} for constant treewidth is optimal under the
ETH.

\begin{theorem}\label{thm:eth1}If the ETH is true then (i) there is no
algorithm solving \msc\ in time $(nW)^{o(\pw)}$ (ii) there is no algorithm
solving \msc\ in time $2^{o(\Delta \pw)}(n+\log W)^{O(1)}$. These statements
apply even if we restrict the input to instances where weights are written in
unary and the maximum degree is $O(\log n)$. \end{theorem}

\begin{proof} 
We recall that the standard chain of reductions from \textsc{3-SAT} to
\textsc{3-Set Splitting} which establishes that the latter problem is NP-hard
produces an instance with size linear in the original formula
\cite{GareyJ79,ImpagliazzoPZ01}. We compose these reductions with the reduction
of Lemma \ref{lem:hard}, setting $\delta=\lceil\log n\rceil$. Suppose we
started with a formula with $n$ variables and $m$ clauses (so as an
intermediate step we constructed a \textsc{3-Set Splitting} instance with
$O(n+m)$ elements and sets). We therefore now have an instance with
$N=poly(n+m)$ vertices (since the reduction runs in polynomial time), maximum
degree $\Delta=O(\log(n+m))$ and pathwidth $\pw=O((n+m)/\log(n+m))$, and
maximum weight $W=poly(n+m)$. Plugging these relations into the running times
of hypothetical algorithms for \msc\ we obtain algorithms for \textsc{3-SAT}
running in time $2^{o(n+m)}$ and contradicting the ETH.  \end{proof}

More interestingly,  because Lemma~\ref{lem:hard} allows us to select an arbitrary
value of $\delta$, we can prove that any FPT algorithm parameterized by
$\Delta+\pw$ must have complexity exponential in the product of the two
parameters, independent of which is smaller than the other.

\begin{theorem}\label{thm:cool} Suppose that there exists an algorithm solving
\msc\ on instances with $n$ vertices and maximum weight $W$ in time
$2^{f(\Delta,\pw)}(n+\log W)^{O(1)}$, for some increasing function $f$.  Let
$\Delta(n)$ be a constructible increasing function on the positive integers
with $\Delta(n)=o(n)$ and $\Delta(n)=\omega(1)$. If for any such function
$\Delta(n)$ we have that $f(\Delta(n),\lceil n/\Delta(n)\rceil)=o(n)$, then the
ETH is false.  \end{theorem}

\begin{proof} We essentially repeat the proof of Theorem~\ref{thm:eth1}, but set
$\delta=\Delta(n)$ rather than $\delta=\log n$. \end{proof}

To demonstrate why Theorem~\ref{thm:cool} provides a stronger hardness result,
consider the examples of the following corollary, which are not ruled out by
Theorem~\ref{thm:eth1} but can be ruled out thanks to the flexibility of
Theorem~\ref{thm:cool}.

\begin{corollary}\label{cor:cool} Let $c\ge 2,\varepsilon>0$ be real constants.
Then, assuming the ETH, \msc\ cannot be solved in time
$2^{O(\Delta^{2-\varepsilon}+\pw^{2-\varepsilon})}(n+\log W)^{O(1)}$, nor in
$2^{O(\Delta^{c}+\pw^{\frac{c}{c-1}-\varepsilon})}(n+\log W)^{O(1)}$, nor in
$2^{O(\Delta^{\frac{c}{c-1}-\varepsilon}+\pw^{c})}(n+\log W)^{O(1)}$.
\end{corollary}

\begin{proof} For the first claim, we invoke Theorem~\ref{thm:cool} with
$\Delta(n)=\sqrt{n}$. If we assume that there is an algorithm for \msc\ with
running time $2^{O(\Delta^{2-\varepsilon}+\pw^{2-\varepsilon})}(n+\log
W)^{O(1)}$, then we have that the function $f$ in Theorem~\ref{thm:cool} is
$f(\Delta,\pw)=\Delta^{2-\varepsilon}+\pw^{2-\varepsilon}$. Then,
$f(\sqrt{n},\sqrt{n}) = O(n^{1-\varepsilon/2})=o(n)$, so this running time is
excluded by Theorem~\ref{thm:cool}.  For the second claim we set
$\Delta(n)=n^{1/c-\varepsilon/4}$ and we have
$f(\Delta(n),n/\Delta(n))=O(n^{1-c\varepsilon/4}+n^{(\frac{c}{c-1}-\varepsilon)(\frac{c-1}{c}+\frac{\varepsilon}{4})})=o(n)$.
To see this, note that we have $\frac{c}{4(c-1)}-\frac{c-1}{c} =
\frac{c^2-4(c-1)^2}{4c(c-1)}$ and $c^2-4(c-1)^2 = -3c^2+8c-4 \le 0$ when $c\ge
2$. The last claim is symmetrically obtained by exchanging the roles of
$\Delta$ and $\pw$. \end{proof}

Thanks to the above corollary we can now much more convincingly argue that the
algorithm of Theorem~\ref{thm:algdelta} is optimal. Take any algorithm where the
exponent of the parameter can additively separate the contribution of $\pw$ and
$\Delta$. If the contribution is at least quadratic in both parameters, then
the running time is higher than $\pw\Delta$, since $\pw^2+\Delta^2>\pw\Delta$.
However, the contribution cannot be sub-quadratic in both parameters, by the
first statement of Theorem~\ref{cor:cool}. Suppose then that the additive contribution
of $\pw$ to the exponent is $\pw^c$, for some $c\ge 2$ (the case where the
contribution of $\Delta$ is at least quadratic is symmetric). Then, by
Corollary~\ref{cor:cool}, the exponent must be at least $\Delta^{\frac{c}{c-1}}+\pw^c$.
We claim that $\Delta^{\frac{c}{c-1}}+\pw^c>\Delta\pw$. To see this, define a
new variable $\alpha$ and set $\pw=\Delta^{\frac{1}{c-1}}\cdot\alpha$. We have
$\Delta^{\frac{c}{c-1}}+\pw^c = \Delta^{\frac{c}{c-1}}(1+\alpha^c)$ while
$\Delta\pw = \Delta^{\frac{c}{c-1}}\alpha$. We claim that $1+\alpha^c>\alpha$
for any $\alpha>0$ and $c\ge 2$, since if $\alpha<1$ the inequality is trivial,
while if $\alpha\ge 1$ we have $\alpha^c\ge\alpha$. We conclude that in all
cases the exponent has to be at least as high as $\Delta\pw$, so the algorithm
of Theorem~\ref{thm:algdelta} is essentially optimal.

\section{Approximately Stable Cuts}

In this section we present an algorithm which runs in FPT time parameterized by
treewidth and produces a solution that is almost stable
($(1+\varepsilon)$-stable, as defined below) and has weight upper bounded by
the weight of the optimal stable cut.  Before we proceed, we will need to
define a more general version of our problem. In \textsc{Extended} \msc\ we are
given as input: a graph $G=(V,E)$; a cut-weight function $w:E\to\mathbb{N}$;
and a stability-weight function $s:E\times V \to \mathbb{N}$. For $v\in V$ we
denote $d_s(v)=\sum_{vu\in E} s(vu,v)$, which we call the stability degree of
$v$. If we are also given an error parameter $\rho>1$, we will then be looking
for a partition of $V$ into $V_0,V_1$ which satisfies the following: (i) each
vertex is $\rho$-stable, that is, for each $i\in\{0,1\}$ and $v\in V_i$ we have
$\sum_{vu\in E\land u\in V_{1-i}} s(vu,v) \ge \frac{d_s(v)}{2\rho}$ (ii) the
total cut weight $\sum_{u\in V_0, v\in V_1, uv\in E} w(uv)$ is minimum. Observe
that this extended version of the problem contains \msc\ as a special case if
$\rho=1$ and for all $uv\in E$ we have $s(uv,v)=s(uv,u)=w(uv)$. 

The generalization of \msc\ is motivated by three considerations.  First, the
algorithm of Theorem \ref{thm:algpseudo} is inefficient because it has to store
exact weight values to satisfy the stability constraints; however, it can
efficiently store the total weight of the cut. We therefore decouple the
contribution of an edge to the size of the cut (given by $w$) from a
contribution of an edge to the stability of its endpoints (given by $s$).
Second, our strategy will be to truncate the values of $s$ so that the DP of
the algorithm of Theorem \ref{thm:algpseudo} can be run more efficiently. To do
this we will first simply divide all stability-weights by an appropriate value.
However, a problem we run into if we do this is that the edge $uv$ could
simultaneously be one of the heavier edges incident on $u$ and one of the
lighter edges incident on $v$, so it is not clear how we can adjust its weight
in a way that minimizes the distortion for both endpoints. As a result it is
simpler if we allow edges to contribute different amounts to the stability of
their endpoints. In this sense, $s(uv,u)$ is the amount that the edge $uv$
contributes to the stability of $u$ if the edge is cut. Observe that with the
new definition, if we set a new stability-weight function for a specific vertex
$u$ as $s'(uv,v) = c\cdot s(uv,v)$ for all $v\in N(u)$, that is, if we multiply
the stability-weight of all edges incident on $u$ by a constant $c$ and leave
all other values unchanged, we obtain an equivalent instance, and this does not
affect the stability of other vertices.  Finally, the parameter $\rho$
allows us to consider solutions where a vertex is stable if its cut incident
edges are at least a $(\frac{1}{2\rho})$-fraction of its stability
degree.

Armed with this intuition we can now explain our approach to obtaining our FPT
approximation algorithm. Given an instance of the extended problem, we first
adjust the $s$ function so that its maximum value is bounded by a polynomial in
$n$. We achieve this by dividing $s(uv,u)$ by a value that depends only on
$d_s(u)$ and $n$. This allows us to guarantee that near-stable solutions are
preserved. Then, given an instance where the maximum value of $s$ is
polynomially bounded, we apply the technique of \cite{Lampis14}, using the
algorithm of Theorem \ref{thm:algpseudo} as a base, to obtain our
approximation. We give these separate steps in the Lemmas below.

\begin{lemma}\label{lem:algapprox1} There is an algorithm which, given a graph
$G=(V,E)$ on $n$ vertices and a stability-weight function $s:E\times
V\to\mathbb{N}$ with maximum value $S$, runs in time polynomial in $n+\log S$
and produces a stability-weight function $s':E\times V\to\mathbb{N}$ with the
following properties: (i) the maximum value of $s'$ is $O(n^2)$ (ii) for all
partitions $V$ into $V_0,V_1$, $i\in\{0,1\}$, $v\in V_i$ we have 

$$ \left| \frac{\sum_{vu\in E, u\in V_{1-i}} s(vu,v)}{d_s(v)} -  \frac{\sum_{vu\in E,
u\in V_{1-i}} s'(vu,v)}{d_{s'}(v)} \right| \le \frac{1}{n} $$  

\end{lemma}

\begin{proof} For $v\in V$ let $S(v) = \max_{u\in N(v)} s(vu,v)$. We define
$s'$ as follows: $s'(vu,v) = \lceil \frac{n^2 s(vu,v)}{S(v)} \rceil$. It is
clear that the maximum value of $s'$ is $n^2$ and that calculations can be
carried out in the promised time. So what remains is to prove that for any
partition the fraction $\frac{\sum_{vu\in E,u\in V_{1-i}} s(vu,v)}{d_s(v)}$
stays essentially unchanged. In particular we will prove that:

$$ \frac{\sum_{vu\in E, u\in V_{1-i}} s(vu,v)}{d_s(v)} - \frac{1}{n} \le
\frac{\sum_{vu\in E, u\in V_{1-i}} s'(vu,v)}{d_{s'}(v)} \le \frac{\sum_{vu\in
E, u\in V_{1-i}} s(vu,v)}{d_s(v)} + \frac{1}{n} $$

Observe that $\frac{n^2 s(vu,v)}{S(v)} \le s'(vu,v) \le \frac{n^2
s(vu,v)}{S(v)}+1$. We therefore have 

$$ \frac{n^2 d_s(v)}{S(v)} \le d_{s'}(v) \le \frac{n^2 d_s(v)}{S(v)}+n$$

We also have:

$$ \frac{n^2 \sum_{vu\in E, u\in V_{1-i}} s(vu,v)}{S(v)} \le \sum_{vu\in E,
u\in V_{1-i}} s'(vu,v) \le \frac{n^2 \sum_{vu\in E, u\in V_{1-i}}
s(vu,v)}{S(v)}+n$$

In both cases we have used the fact that the degree of $v$ is at most $n$. Now
with some calculation we get:

$$ \frac{\sum_{vu\in E, u\in V_{1-i}} s(vu,v)}{d_s(v) + \frac{S(v)}{n}}\le
\frac{\sum_{vu\in E, u\in V_{1-i}} s'(vu,v)}{d_{s'}(v)} \le \frac{\sum_{vu\in
E, u\in V_{1-i}} s(vu,v) + \frac{S(v)}{n}}{d_s(v)} $$

Now, from the second part of the above inequality we get:

$$ \frac{\sum_{vu\in E, u\in V_{1-i}} s'(vu,v)}{d_{s'}(v)} \le \frac{\sum_{vu\in
E, u\in V_{1-i}} s(vu,v)}{d_s(v)}+ \frac{S(v)}{nd_s(v)} \le \frac{\sum_{vu\in
E, u\in V_{1-i}} s(vu,v)}{d_s(v)}+ \frac{1}{n} $$

where we took into account that $S(v)\le d_s(v)$. On the other hand, taking
into account that $\frac{1}{1+\frac{x}{n}}>1-\frac{x}{n}$ for all positive
$x,n$, the first part of the inequality gives:

$$ \frac{\sum_{vu\in E, u\in V_{1-i}} s'(vu,v)}{d_{s'}(v)} \ge
\frac{\sum_{vu\in E, u\in V_{1-i}} s(vu,v)}{d_s(v)(1 + \frac{S(v)}{nd_s(v)})}
\ge \frac{\sum_{vu\in E, u\in V_{1-i}} s(vu,v)}{d_s(v)} \cdot
(1-\frac{S(v)}{nd_s(v)}) $$

Again using that $S(v)\le d_s(v)$ and that $\frac{\sum_{vu\in E, u\in V_{1-i}}
s(vu,v)}{d_s(v)} \le 1$ we get $ \frac{\sum_{vu\in E, u\in V_{1-i}}
s'(vu,v)}{d_{s'}(v)} \ge \frac{\sum_{vu\in E, u\in V_{1-i}} s(vu,v)}{d_s(v)} -
\frac{1}{n}$.  \end{proof}

Using Lemma \ref{lem:algapprox1} we can assume that all stability-weights are
bounded by $O(n^2)$. The most important part is that Lemma \ref{lem:algapprox1}
guarantees us that almost-optimal solutions are preserved in both directions,
as for any cut and for each vertex the ratio of stability weight going to the
other side over the total stability-degree of the vertex does not change by
more than an additive term of $\frac{1}{n}$.  Let us now see the second
ingredient of our algorithm.

\begin{lemma}\label{lem:algapprox2} There is an algorithm which takes as input
a graph $G=(V,E)$, a cut-weight function $w:E\to\mathbb{N}$ with maximum $W$, a
stability-weight function $s:E\times V\to \mathbb{N}$ with maximum $S$, a tree
decomposition of $G$ of width $\tw$, and an error parameter $\varepsilon\in
(0,1)$ and returns a $(1+2\varepsilon)$-stable solution that has cut-weight at
most equal to that of the minimum $(1+\varepsilon)$-stable solution.  If
$S=O(n^2)$, then the algorithm runs in time $(\tw/\varepsilon)^{O(\tw)}(n+\log
W)^{O(1)}$.  \end{lemma}

\begin{proof} 
We use the methodology of \cite{Lampis14}. We assume we are given a nice tree
decomposition of height $H$, where the height of a decomposition is defined as
the longest distance from the root to a leaf. We will formulate an algorithm
with running time $(H\log S/\varepsilon)^{O(\tw)}(n+\log W)^{O(1)}$. We explain
in the end why this is sufficient to obtain the claimed running time, which
does not explicitly depend on $H$ or $S$.

The basis of our algorithm will be the algorithm of Theorem
\ref{thm:algpseudo}, appropriately adjusted to the extended version of the
problem. Let us first sketch the modifications to the algorithm of Theorem
\ref{thm:algpseudo} that we would need to do to solve this more general
problem, since the details are straightforward. First, we observe that in
solution signatures we would now take into account stability-weights, and
signatures would have values going up to $\Delta S$, where $\Delta$ is the
maximum (unweighted) degree.  Second, in Forget nodes, recall that the
algorithm of Theorem \ref{thm:algpseudo} discards (skips) signatures which
correspond to unstable solutions because the vertex we are forgetting has too
large weight towards its own side.  Since we are happy with a
$(1+\varepsilon)$-stable solution, we only discard solutions which violate this
constraint, that is, where the vertex we forget has a clear majority of its
incident weight (more than a $(1+\varepsilon)\frac{1}{2}$ fraction) to its own
side.  With these modifications, we can run this exact algorithm to return the
minimum $(1+\varepsilon)$-stable solution in time $(\Delta S)^{O(\tw)}(n+\log
W+\log(1/\varepsilon))^{O(1)}$.

The idea is to modify this algorithm so that the DP tables go from size
$(2\Delta S)^{\tw}$ to roughly $(H\log S)^{\tw}$. To do this, we define a
parameter $\delta = \frac{\varepsilon}{5H}$. We intend to replace every value
$x$ that would be stored in the signature of a solution in the DP table, with
the next larger integer power of $(1+\delta)$, that is, to construct a DP table
where $x$ is replaced by $(1+\delta)^{\lceil \log_{(1+\delta)} x \rceil}$.

More precisely, the invariant we maintain is the following. Consider a node $t$
of the decomposition at height $h$, where $h=0$ corresponds to leaves. We
maintain a collection of solution signatures such that: (i) each signature
contains a partition of $B_t$ and for each $v\in B_t$ an integer that is
upper-bounded by $\lceil\log_{(1+\delta)} d_s(v)\rceil$; (ii) Soundness: for
each stored signature there exists a partition of $B^{\downarrow}_t$ which
approximately corresponds to it. Specifically, the partition and the signature
agree exactly on the assignment of $B_t$ and the total cut-weight; the
partition is $(1+2\varepsilon)$-stable for all vertices of
$B^{\downarrow}_t\setminus B_t$; and for each $v\in B_t$, if the signature
stores the value $x(v)$ for $v$, that is, it states that $v$ has approximate
stability-weight $(1+\delta)^{x(v)}$ towards its own side in
$B^{\downarrow}_t\setminus B_t$, then in the actual partition the
stability-weight of $v$ to its own side of $B^{\downarrow}_t\setminus B_t$ is
at most $(1+\delta)^h(1+\delta)^{x(v)}$. (iii) Completeness: conversely, for
each partition of $B^{\downarrow}_t$ that is $(1+\varepsilon)$-stable for all
vertices of $B^{\downarrow}_t\setminus B_t$ there exists a signature that
approximately corresponds to it. Specifically, the partition and signature
agree on the assignment of $B_t$ and the total cut-weight; and for each $v\in
B_t$, if the stability-weight of $v$ towards its side of the partition of
$B^{\downarrow}_t\setminus B_t$ is $y(v)$, and the signature stores the value
$x(v)$, then $(1+\delta)^{x(v)}\le (1+\delta)^h y(v)$.

In more simple terms, the signatures in our DP table store values $x(v)$ so
that we estimate that in the corresponding solution $v$ has approximately
$(1+\delta)^{x(v)}$ weight towards its own side in $B_t^{\downarrow}$, that is,
we estimate that the DP of the exact algorithm would store approximately the
value $(1+\delta)^{x(v)}$ for this solution. Of course, it is hard to maintain
this relation exactly, so we are happy if for a node at height $h$ the ``true''
value which we are approximating is at most a factor of $(1+\delta)^h$ off from
our approximation.

Now, the crucial observation is that the approximate DP tables can be
maintained because our invariant allows the error to increase with the height.
For example, suppose that $t$ is a Forget node at height $h$ and let $u\in B_t$
be a neighbor of the vertex $v$ we forget. The exact algorithm would construct
the signature of a solution in $t$ by looking at the signature of a solution in
its child node, and then adding to the value stored for $u$ the weight
$s(vu,u)$ (if $u,v$ are on the same side). Our algorithm will take an
approximate signature from the child node, which may have a value at most
$(1+\delta)^{h-1}$ the correct value, add to it $s(vu,u)$ and then, perhaps,
round-up the value to an integer power of $(1+\delta)$. The new approximation
will be at most $(1+\delta)^h$ larger than the value that the exact algorithm
would have calculated. To give more details, if the ``correct'' value that the
exact algorithm would have calculated at a node is $x$, our approximation
algorithm aims to store a value $\hat{x}$ satisfying $x\le \hat{x}\le
(1+\delta)^{h-1}x$. Where the exact algorithm would have calculated a new value
$z:=x+y$ (where $y$ is the weight of some edge incident on the vertex we
forget), our algorithm calculates $\hat{x}+y$, for which it clearly holds that
$z=x+y \le \hat{x}+y \le (1+\delta)^{h-1}(x+y)$.  However, because $\hat{x}+y$
is (probably) not an integer power of $(1+\delta)$, our algorithm will round up
$\hat{x}+y$ into a new value $\hat{z}$. Since rounding up in this way can
increase the stored value by at most a factor $(1+\delta)$, we have $z\le
\hat{z}\le (1+\delta)^hz$, so by induction all the approximate values we
compute are lower bounded by the correct values and are at most a factor
$(1+\delta)^h$ away from them.

Similar argumentation holds for Join nodes. Furthermore, in Forget nodes we
will only discard a solution if according to our approximation it is not
$(1+2\varepsilon)$-stable. We may be over-estimating the stability-weight a
vertex has to its own side of the cut by a factor of at most $(1+\delta)^h \le
(1+\frac{\varepsilon}{5H})^H \le 1+\frac{\varepsilon}{2}$ so if for a signature
our approximation says that the solution is not $(1+2\varepsilon)$-stable, the
solution cannot be $(1+\varepsilon)$-stable, because
$(1+\varepsilon)(1+\frac{\varepsilon}{2})<1+2\varepsilon$ (for
$\varepsilon<1$).  

Finally, to estimate the running time, the maximum value we have to store for
each vertex in a bag is $\log_{(1+\delta)} (\Delta S) \le \frac{\log n S}{\log
(1+\delta)} \le O(\frac{\log n}{\delta}) = O(\frac{H\log n}{\varepsilon})$. The
size of the DP table is therefore $(H\log n/\varepsilon)^{O(\tw)}$ and the
running time is polynomial in this size and the input size.

Let us now explain how we go from the above bound to the running time stated.
First, a theorem due to \cite{BodlaenderH98} proves that any tree decomposition
can be edited (in polynomial time) so that its height becomes $O(\log n)$,
without increasing its width by more than a constant factor. Further editing
the decomposition so that it becomes nice may further increase the height by a
factor of $\tw$, so we can assume that $H=O(\tw\log n)$. Therefore, we have a
running time of the form $(\tw\log n/\varepsilon)^{O(\tw)}$.

Finally, to obtain an upper bound on the running time that matches the promised
bound, we use a standard Win/Win argument. First, suppose that  $\tw\le
\sqrt{\log n}$. In this case, $(\tw\log n/\varepsilon)^{O(\tw)} = n^{o(1)}$,
and the whole algorithm runs in polynomial time. So the interesting case is
when $\log n\le \tw^2$. But then, the running time can be bounded by
$(\tw/\varepsilon)^{O(\tw)}(n+\log W)^{O(1)}$, as promised.  \end{proof}

\begin{theorem}\label{thm:algapprox} There is an algorithm which, given an
instance of \msc\ $G=(V,E)$ with $n$ vertices, maximum weight $W$, a tree
decomposition of width $\tw$, and a desired error $\varepsilon>0$, runs in time
$(\tw/\varepsilon)^{O(\tw)}(n+\log W)^{O(1)}$ and returns a cut with the
following properties: (i) for all $v\in V$, the total weight of edges incident
on $v$ crossing the cut is at least $(1-\varepsilon)\frac{d_w(v)}{2}$ (ii) the
cut has total weight at most equal to the weight of the minimum stable cut.
\end{theorem}

\begin{proof}
We simply put together the algorithms of Lemmas \ref{lem:algapprox1} and
\ref{lem:algapprox2}. Fix an $\varepsilon>0$. If $n<\frac{100}{\varepsilon}$,
then we can try out all partitions in constant time (depending only on
$\varepsilon$), so assume that $n$ is larger than that. We execute the
algorithm of Lemma \ref{lem:algapprox1} so the weight of all cuts is preserved
(since we do not change $w$), and a stable cut remains at least
$(1+\varepsilon/2)$-stable, since $n$ is sufficiently large.  We therefore
execute the algorithm of Lemma \ref{lem:algapprox2} and this will output a
$(1+\varepsilon)$-stable cut with value at least as small as the minimum stable
cut.  \end{proof}

\section{Unweighted Min Stable Cut}

In this section we consider \umsc. We first observe that applying Theorem
\ref{thm:algpseudo} gives a parameter dependence of $\Delta^{O(\tw)}$, since
$W=1$.  We then show that this algorithm is essentially optimal, as the problem
cannot be solved in $n^{o(\pw)}$ under the ETH.

\begin{corollary}\label{cor:algdelta2} There is an algorithm which, given an
instance of \umsc\ with $n$ vertices, maximum degree $\Delta$, and a tree
decomposition of width $\tw$, returns an optimal solution in time
$\Delta^{O(\tw)} n^{O(1)}$.  \end{corollary}

\begin{wrapfigure}{R}{0.3\textwidth}

\input{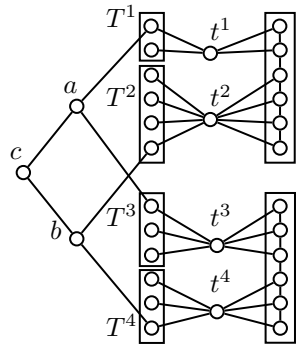}
\caption{Checker gadget for Theorem \ref{thm:hard2}. On the right two Selector gadgets. This Checker verifies that we have not taken an edge which has endpoints $(2,3)$, hence $t^1,t^3$ are connected to the first $2$ and $3$ vertices of the Selectors.}\label{fig:red2}

\end{wrapfigure}

We now first state our hardness result, then describe the construction of our
reduction, and finally go through a series of lemmas that establish its
correctness.

\begin{theorem}\label{thm:hard2} If the ETH is true then no algorithm can solve
\umsc\ on graphs with $n$ vertices in time $n^{o(\pw)}$. Furthermore, \umsc\ is
W[1]-hard parameterized by pathwidth. \end{theorem}

To prove Theorem \ref{thm:hard2} we will describe a reduction from
$k$-\textsc{Multi-Colored Independent Set}, a well-known W[1]-hard problem that
cannot be solved in $n^{o(k)}$ time under the ETH \cite{CyganFKLMPPS15}. In
this problem we are given a graph $G=(V,E)$ with $V$ partitioned into $k$ color
classes $V_1,\ldots, V_k$, each of size $n$, and we are asked to find an
independent set of size $k$ which selects one vertex from each $V_i$. In the
remainder we use $m$ to denote the number of edges of $E$ and assume that
vertices of $V$ are labeled $v_{(i,j)}, i\in [k], j\in [n]$, where $V_i = \{
v_{(i,j)}\ |\ j\in [n]\}$.

Before we proceed, let us give some intuition. Our reduction will rely on a
$k\times m$ grid-like construction, where each row represents the selection of
a vertex in the corresponding color class of $G$ and each column represents an
edge of $G$. The main ingredients will be a Selector gadget, which will
represent a choice of an index in $[n]$; a Propagator gadget which will make
sure that the choice we make in each row stays consistent throughout; and a
Checker gadget which will verify that we did not select the two endpoints of
any edge. Each Selector gadget will contain a path on (roughly) $n$ vertices
such that any reasonable stable cut will have to cut exactly one edge of the
path. The choice of where to cut this path will represent an index in $[n]$
encoding a vertex of $G$.

In our construction we will also make use of a simple but important gadget
which we will call a ``heavy'' edge. Let $A=n^5$. When we say that we connect
$u,v$ with a heavy edge we will mean that we construct $A$ new vertices and
connect them to both $u$ and $v$. The intuitive idea behind this gadget is that
the large number of degree two vertices will force $u$ and $v$ to be on
different sides of the partition (otherwise too many edges will be cut). We
will also sometimes attach leaves on some vertices with the intention of making
it easier for this vertex to achieve stability (as its attached leaves will
always be on the other side of the partition).

Let us now describe our construction step-by-step.

\begin{enumerate}

\item Construct two ``palette'' vertices $p_0, p_1$ and a heavy edge connecting
them. Note that all heavy edges we will add will be incident on at least one
palette vertex.

\item For each $i\in [k], j\in [m]$ construct the following Selector gadget:

	\begin{enumerate}

	\item Construct a path on $n+1$ vertices $P_{(i,j)}$ and label its
vertices $P_{(i,j)}^1, \ldots, P_{(i,j)}^{n+1}$.

	\item If $j$ is odd, then add a heavy edge from $P_{(i,j)}^1$ to $p_1$
and a heavy edge from $P_{(i,j)}^{n+1}$ to $p_0$. If $j$ is even, then add a
heavy edge from $P_{(i,j)}^1$ to $p_0$ and a heavy edge from $P_{(i,j)}^{n+1}$
to $p_1$.

	\item Attach $5$ leaves to each $P_{(i,j)}^{\ell}$ for $\ell\in
\{2,\ldots,n\}$. Attach $A+5$ leaves to $P_{(i,j)}^1$ and $P_{(i,j)}^{n+1}$.

	\end{enumerate}

\item For each $i\in [k], j\in [m-1]$ construct a new vertex connected to all
vertices of the paths $P_{(i,j)}$ and $P_{(i,j+1)}$. This vertex is the
Propagator gadget.

\item For each $j\in [m]$ consider the $j$-th edge of the original instance and
suppose it connects $v_{(i_1,j_1)}$ to $v_{(i_2,j_2)}$. We construct the
following Checker gadget (see Figure \ref{fig:red2})

	\begin{enumerate}

	\item We construct four vertices $t_j^1, t_j^2, t_j^3, t_j^4$. These
are connected to existing vertices as follows: $t_j^1$ is connected to
$\{P_{(i_1,j)}^1, \ldots, P_{(i_1,j)}^{j_1}\}$ (that is, the first $j_1$
vertices of the path $P_{(i_1,j)}$); $t_j^2$ is connected to
$\{P_{(i_1,j)}^{j_1+1}, \ldots, P_{(i_1,j)}^{n+1}\}$ (that is, the remaining
$n+1-j_1$ vertices of $P_{i_1,j}$); similarly, $t_j^3$ is connected to
$\{P_{(i_2,j)}^1, \ldots, P_{(i_2,j)}^{j_2}\}$; and finally $t_j^4$ is
connected to $\{P_{(i_2,j)}^{j_2+1}, \ldots, P_{(i_2,j)}^{n+1}\}$.

	\item We construct four independent sets $T_j^1, T_j^2, T_j^3, T_j^4$
with respective sizes $j_1, n+1-j_1, j_2, n+1-j_2$. We connect $t_j^1$ to all
vertices of $T_j^1$, $t_j^2$ to $T_j^2$, $t_j^3$ to $T_j^3$, and $t_j^4$ to
$T_j^4$. We attach two leaves to each vertex of $T_j^1\cup T_j^2\cup T_j^3\cup
T_j^4$.

	\item We construct three vertices $a_j, b_j, c_j$. We connect $c_j$ to
both $a_j$ and $b_j$. We connect $a_j$ to an arbitrary vertex of $T_j^1$ and an
arbitrary vertex of $T_j^3$. We connect $b_j$ to an arbitrary vertex of $T_j^2$
and an arbitrary vertex of $T_j^4$.

	\end{enumerate}

\end{enumerate}

Let $L_1$ be the number of leaves of the construction we described above and
$L_2$ be the number of degree two vertices which are part of heavy edges. We
set $B=L_1 + L_2 + km + k(m-1)(n+1) + m(2n+6)$.

\begin{lemma}\label{lem:red2a} If $G$ has a multi-colored independent set of
size $k$, then the constructed instance has a stable cut of size  $B$.
\end{lemma}

\begin{proof}
Let $\sigma: [k] \to [n]$ be a function that encodes a multi-colored
independent set of $G$, that is, the set $\{ v_{(i,\sigma(i))}\ |\ i\in [k]\}$
is an independent set. We construct a partition of the new instance as follows:
we assign $0$ to $p_0$, $1$ to $p_1$, and arbitrary values to the vertices of
the heavy edge connecting $p_0$ to $p_1$; each other vertex that belongs to a
heavy edge incident to $p_0$ (respectively $p_1$) is assigned $1$ (respectively
$0$); each vertex connected via a heavy edge to $p_0$ (respectively $p_1$) is
assigned $1$ (respectively $0$); for each Selector gadget $P_{(i,j)}$ we assign
to the first $\sigma(i)$ vertices of the path (that is, the vertices
$\{P_{i,j}^1,\ldots, P_{i,j}^{\sigma(i)}\}$) the same value as $P_{i,j}^1$
(that is, $0$ if $j$ is odd and $1$ if $j$ is even); we assign to the remaining
vertices of $P_{(i,j)}$ the same value as $P_{(i,j)}^{n+1}$; we assign to every
leaf the opposite value from that of its neighbor; we assign an arbitrary value
to each Propagator vertex. We have now described a partition of all the
vertices except of the non-leaf vertices belonging to Checker gadgets.

Before we describe the partition of the Checker gadgets let us establish some
basic properties of the partition so far. First, all vertices for which we have
given a value are stable, independent of the values we intend to assign to the
non-leaf Checker gadget vertices.  To see this we note that (i) all leaves have
a value different from their neighbors (ii) all degree $2$ vertices that belong
to heavy edges have two neighbors with distinct values (iii) $p_0$ and $p_1$
have the majority of their neighbors on the other side of the partition (iv)
for all non-leaf Selector gadget vertices at least half their neighbors are
leaves (which are on the opposite side of the partition) (v) all Propagator
vertices have exactly $n+1$ neighbors on each side of the partition.  The total
number of edges cut so far is (i) $L_1$ edges incident on leaves (ii) $L_2$
edges incident on degree $2$ vertices that belong to heavy edges (iii) one
internal edge of each path $P_{(i,j)}$ giving $km$ edges in total (iv) half of
the $2n+2$ edges incident on each Propagator vertex, of which there are
$k(m-1)$, giving $k(m-1)(n+1)$ in total.  Summing up, we have already cut $L_1
+ L_2 + km + k(m-1)(n+1)$ edges, meaning we can still cut $m(2n+6)$ edges. We
will describe a stable partition of the Checker gadgets which cuts exactly
$2n+6$ edges per gadget (not counting edges incident on leaves, since these are
already counted in $L_1$), and since we have $m$ Checker gadgets this will
complete the proof.

Consider now the Checker gadget for edge $j$ which connects $v_{(i_1,j_1)}$ to
$v_{(i_2,j_2)}$ and without loss of generality assume that $j$ is odd
(otherwise the proof is identical with the roles of $0$ and $1$ reversed). We
claim that one of the vertices $t_j^1, t_j^2, t_j^3, t_j^4$ must have neighbors
on both sides of the partition in the Selector gadgets. To see this, suppose
for contradiction that each of these vertices only has neighbors on one side of
the partition so far. Then, since $t_j^1$ is connected to $P_{(i_1,j)}^1$,
which has color $0$ and $t_j^2$ is connected to $P_{(i_1,j)}^{n+1}$, which has
color $1$, and $t_j^1$ is connected to the first $j_1$ vertices of
$P_{(i_1,j)}$, we conclude that $\sigma(i_1)=j_1$, because the number of
vertices of the path $P_{(i_1,j)}$ which have value $0$ is $\sigma(i_1)$.  With
the same argument, we must have $\sigma(i_2)=j_2$, contradicting the hypothesis
that $\sigma$ encodes an independent set.

We can therefore assume that one of $t_j^1, t_j^2, t_j^3, t_j^4$ has neighbors
on both sides of the partition in the Selector gadgets. Without loss of
generality suppose that $t_j^1$ has this property (the proof is symmetric in
other cases).  We complete the partition as follows: we assign values to
$T_j^2, T_j^3, T_j^4$ in a way that $t_j^2, t_j^3, t_j^4$ have the same number
of neighbors on each side of the partition and that both neighbors of $b_j$ in
$T_j^2, T_j^4$ have value $0$.  This is always possible as $t_j^2, t_j^4$ have
a neighbor with value $1$ in the Selectors, namely $P_{(i_1,j)}^{n+1}$ and
$P_{(i_2,j)}^{n+1}$.  We assign colors to $T_j^1$ in a way that $t_j^1$ has the
same number of neighbors on each side and $a_j$ has two neighbors with distinct
values in $T_j^1\cup T_j^3$. This is always possible as we need to use both
values in $T_j^1$, because $t_j^1$ has neighbors with both values in
$P_{(i_1,j)}$. We give $b_j$ value $1$, $c_j$ value $1$ and $a_j$ value $0$.
This is stable as $b_j$ has two neighbors of value $0$, $c_j$ has neighbors
with distinct values, and $a_j$ has two neighbors with value $1$. Furthermore,
vertices in $T_j^1\cup T_j^2\cup T_j^3\cup T_j^4$ are stable because at
least half their neighbors are leaves which are on the other side of the
partition, and the neighborhoods of $t_j^1, t_j^2, t_j^3, t_j^4$ are completely
balanced, so these vertices can be arbitrarily set.  The number of edges cut is
half of the edges incident on $t_j^1, t_j^2, t_j^3, t_j^4$, giving $2n+2$
edges, plus two edges incident on each of $a_j, b_j$, giving a total of $2n+6$
edges.  \end{proof}

\begin{lemma}\label{lem:red2b} If the constructed instance has a stable cut of size at most $B$,
then $G$ has a multi-colored independent set of size $k$. \end{lemma}

\begin{proof}
Suppose we have a stable cut of size at most $B$. This cut must include all
$L_1$ edges incident on leaves, and at least one edge for each of the $L_2$
degree two vertices which belong to heavy edges. Furthermore, if there is a
heavy edge such that both of its endpoints have the same value, the number of
edges cut incident on vertices that belong to heavy edges will be at least
$L_2+A$. However, $A=n^5> km +k(m-1)(n+1) + m(2n+6)$, so we would have a cut of
size strictly larger than $B$. We conclude that in all heavy edges the two
endpoints have distinct values. Without loss of generality assume value $0$ is
given to $p_0$ and $1$ to $p_1$.

We now observe that:

\begin{enumerate}

\item At least one internal edge of each path $P_{(i,j)}$ is cut.

\item At least $n+1$ edges incident on each Propagator vertex are cut.

\item At least $2n+6$ edges not incident to leaves are cut inside each Checker
gadget.

\end{enumerate}

For the first claim, observe that if the endpoints of heavy edges take distinct
values, this implies that in each path $P_{(i,j)}$ the first and last vertex
have distinct values, so at least one edge of the path must be cut. The second
claim is based on the fact that Propagator vertices have degree $2n+2$. For the
third claim, observe that $t_j^1, t_j^2, t_j^3, t_j^4$ have $4n+4$ edges
incident on them, so at least $2n+2$ of these must be cut in a stable solution.
Furthermore, $a_j, b_j$ have degree $3$, so at least $2$ edges incident on each
of these vertices are cut, giving a total of $2n+6$. (Here, we used the fact
that $\{t_j^1, t_j^2, t_j^3, t_j^4, a_j, b_j\}$ is an independent set).

By the above observations we have that any stable cut must have size at least
$L_1 + L_2 + km + k(m-1)(n+1) + m(2n+6)=B$. Furthermore, if a solution cuts
more than one edge of a path $P_{(i,j)}$, or at least $n+2$ edges incident on a
Propagator, or at least $2n+7$ edges not incident to leaves in a Checker, then
its total size must be strictly larger than $B$. We conclude that our solution
must cut exactly one edge inside each Selector, properly balance the
neighborhoods of all Propagators, and cut $2n+6$ edges inside each Checker.

Consider now two consecutive Selector gadgets $P_{(i,j)}$ and $P_{(i,j+1)}$.
Since the solution cuts exactly one internal edge of each path, we can assume
that the first $x$ vertices of $P_{(i,j)}$ have the same value as $P_{(i,j)}^1$
and the remaining $n+1-x$ have the same value as $P_{(i,j)}^{n+1}$. Similarly,
the first $y$ vertices of $P_{(i,j+1)}$ have the same value as $P_{(i,j+1)}^1$.
Now, because $j, j+1$ have different parities, this means that the Propagator
connected to these two paths has $n+1-x+y$ neighbors on the same side as
$P_{(i,j)}^{n+1}$. But this implies that $x=y$. Using the same reasoning we
conclude that for all $i,j,j'$, the number of vertices of $P_{(i,j)}$ that
share the value of $P_{(i,j)}^1$ is equal to the number of vertices of
$P_{(i,j')}$ that share the value of $P_{(i,j')}^1$. Let $\sigma(i)$ be the
number of vertices of $P_{(i,1)}$ which share the value of $P_{(i,1)}^1$. We
claim that $\{ v_{(i,\sigma(i))}\ |\ i\in [k]\}$ is an independent set in $G$.

To see this, suppose for contradiction that the $j$-th edge of $G$ connects
$v_{(i_1,\sigma(i_1))}$ to $v_{(i_2,\sigma(i_2))}$. We claim that in this case
the Checker connected to $P_{(i_1,j)}, P_{(i_2,j)}$ will have at least $2n+7$
cut edges. Indeed, consider the neighborhood of $t_j^1$ in the Selector gadget
$P_{(i_1,j)}$ and observe that (i) by construction, $t_j^1$ is adjacent to the
first $\sigma(i_1)$ vertices of the path $P_{(i_1,j)}$ (ii) by the definition
of $\sigma(i_1)$, these are exactly the vertices of the path which are on the
same side as $P_{(i_1,j)}^1$. By similar reasoning for the other vertices, we
conclude that the neighborhoods of $t_j^1, t_j^2, t_j^3, t_j^4$ in
$P_{(i_1,j)}$ and $P_{(i_2,j)}$ are \emph{uniform}, that is, the neighborhood
of each such vertex in the Selectors is fully contained on one side of the
partition.

Recall now that by our previous counting, if the neighborhood of one of the
vertices $t_j^1, t_j^2, t_j^3, t_j^4$ is not completely balanced, then at least
$2n+7$ edges will be cut inside the gadget and we are done. Suppose then that
the four neighborhoods are in fact completely balanced. By the reasoning of the
previous paragraph we conclude that each of the sets $T_j^1, T_j^2, T_j^3,
T_j^4$ is fully contained on one of the two sides of the partition and
furthermore, $T_j^1\cup T_j^3$ are on one side and $T_j^2\cup T_j^4$ are on the
other. This implies that $a_j, b_j$ must be on distinct sides of the partition.
As a result, no matter where $c_j$ is placed, one of $a_j,b_j$ will have all
three of its incident edges cut and as a result at least $2n+7$ edges will be
cut in this Checker. We conclude that $\sigma$ must encode an independent set.
\end{proof}

\begin{lemma}\label{lem:pw} The constructed graph has pathwidth $O(k)$. \end{lemma}

\begin{proof}
We will use the fact that deleting a vertex from a graph can decrease the
pathwidth by at most $1$, since we can take a path decomposition of the
resulting graph and add this vertex to all bags. We begin by deleting $p_0,
p_1$ from the graph, as this decreases the pathwidth by at most $2$. We will
also use the fact that deleting all leaves from a graph can decrease pathwidth
by at most $1$, since we can take a path decomposition of the resulting graph
and, for each leaf, find a bag of this decomposition that contains the leaf's
neighbor and insert a copy of this bag immediately after it, adding the leaf.
We therefore remove all leaves from the graph, decreasing the pathwidth by at
most $1$ more.  Let $H$ be the resulting graph. We will show that $H$ has
pathwidth at most $O(k)$.  Observe that in $H$ all heavy edges have
disappeared, as their internal vertices became leaves when we deleted $p_0,
p_1$.

For $j\in [m]$ let $H_j$ be the graph induced by the set that contains all
vertices of $H$ from Selector gadgets $P_{(i,j)}$ for $i\in [k]$,  the (at most
$2k$) Propagator vertices connected to them, and the Checker gadget for the
$j$-th edge.  We will construct a path decomposition of $H_j$ with the property
that all bags include all Propagator vertices of $H_j$.  If we achieve this
then we can make a path decomposition of $H$ by gluing together these
decompositions, connecting the last bag of the decomposition of $H_j$ with the
first bag of the decomposition of $H_{j+1}$.  Observe that the union of the
graphs $H_j$ covers all vertices and edges of $H$.

To build such a path decomposition of $H_j$ we can remove the $2k$ Propagators
contained in $H_j$ (since we will add them in all bags) and the vertices
$t_j^1, t_j^2, t_j^3, t_j^4, a_j, b_j$, decreasing pathwidth by at most $2k+6$.
But the resulting graph is a union of paths and isolated vertices, so has
pathwidth $1$. We can therefore build a decomposition of $H_j$ -- and by
extension of $H$ -- of width $2k+O(1)$.  \end{proof}

\section{Conclusions}

Our results paint a clear picture of the complexity of \msc\ with respect to
$\tw$ and $\Delta$. As directions for further work one could consider stronger
notions of stability such as demanding that switching sets of $k$ vertices
cannot increase the cut, for constant $k$.  We conjecture that, since the
structure of this problem has the form $\exists \forall_k$, its complexity with
respect to treewidth will turn out to be double-exponential in $k$
\cite{LampisM17}.

\bibliographystyle{abbrvnat}
\bibliography{minstablecut}
\label{sec:biblio}

\end{document}